\documentclass[aps,prd,twocolumn, groupedaddress,amssymb,eqsecnum,showpacs,10pt,showkeys,epsfig,nofootinbib]{revtex4}
\usepackage{epsfig}
\usepackage{graphics}
\usepackage{bm}
\usepackage{color}
\usepackage{dcolumn}   
\usepackage[spanish,english]{babel}
\usepackage{bm}     
\usepackage{bbm}       
\usepackage{amssymb}  
\usepackage{amsmath}
\usepackage{latexsym}
\usepackage{float}
\usepackage{ifthen}
\usepackage{caption,subfig}
\usepackage{enumerate}
\usepackage{url}
\usepackage{caption,subfig}
\usepackage{amsopn}
\usepackage{hyperref}
\usepackage[utf8]{inputenc}
\usepackage[english]{babel}
\usepackage{amsthm}
\usepackage[utf8]{inputenc}
\usepackage[english]{babel}

\newtheorem{theorem}{Theorem}[section]
\newtheorem{corollary}{Corollary}[theorem]
\newtheorem{lemma}[theorem]{Lemma}

\def \lleq {\lower0.9ex\hbox{ $\buildrel < \over \sim$} ~}
\def \ggeq {\lower0.9ex\hbox{ $\buildrel > \over \sim$} ~}

\def \beq  {\begin{equation}}
\def \eeq  {\end{equation}}
\def \ber  {\begin{eqnarray}}
\def \eer  {\end{eqnarray}}

\begin{document}
\newcommand{\newc}{\newcommand}
\newc{\ben}{\begin{eqnarray}}
\newc{\een}{\end{eqnarray}}
\newc{\be}{\begin{equation}}
\newc{\ee}{\end{equation}}
\newc{\ba}{\begin{eqnarray}}
\newc{\ea}{\end{eqnarray}}
\newc{\bea}{\begin{eqnarray*}}
\newc{\eea}{\end{eqnarray*}}
\newc{\D}{\partial}
\newc{\na}{\nabla}
\newc{\ie}{{\it i.e.} }
\newc{\eg}{{\it e.g.} }
\newc{\etc}{{\it etc.} }
\newc{\etal}{{\it et al.}}
\newcommand{\nn}{\nonumber}
\newc{\ra}{\rightarrow}
\newc{\lra}{\leftrightarrow}
\newc{\lsim}{\buildrel{<}\over{\sim}}
\newc{\gsim}{\buildrel{>}\over{\sim}}
\newc{\G}{\Gamma}
\title{ Bose-Einstein Condensation for an Exponential Density of States Function and Lerch Zeta Function }
 
 \author{Davood Momeni
 	$^{1,2,3}$ }
\affiliation{\it $^1$ Center for Space Research, North-West University, Mafikeng, South Africa
	\\
	$^2$Tomsk State Pedagogical University, TSPU, 634061 Tomsk, Russia\\
	$^3$ 
	Department of Physics, College of Science, Sultan Qaboos University,\\P.O. Box 36, ,AL-Khodh 123  Muscat, Sultanate of Oman 
	\\  
}
 
\date{\today}

\begin{abstract} 
I show how Bose-Einstein condensation (BEC) in a non interacting  bosonic system with exponential density of states function yields to   a new class of Lerch zeta functions. By looking on the critical temperature, I suggest that a possible strategy to prove the "Riemann hypothesis" problem. In a theorem and a lemma I suggested that the  classical limit $\hbar\to 0$ of   BEC can be used  as a tool to find zeros of real part of the Riemann zeta function with complex argument. It reduces the  Riemann hypothesis to a softer form. Furthermore I propose a pair  of creation-annihilation operators  for BEC phenomena. This set  of creation-annihilation operators is defined on a complex Hilbert space. They build a set up to interpret this type of BEC as a  creation-annihilation phenomenon for a virtual hypothetical particle. 
\end{abstract}

\pacs{03.75.Nt,02.30.Gp, 02.70.Rr}
\keywords{General statistical methods, Bose-Einstein condensation, special functions of mathematical physics, Lerch zeta function, Riemann hypothesis }

\maketitle
\section{Introduction}
In statistical physics there are many examples of critical phenomena where the system undergoes phase transtions below speciefic values of temperature/density or other physical quantities. A remarkable example is the BEC which was finally demonstrated experimentally in dilute alkali gases by a group of reseachers and brought  2001 Nobel Prize in physics \cite{observation}. From a physical point of view the system in the BEC phase exists at the very low temperature regime $T$. In bosonic systems , all bosons 
 condense to the  ground state with the lowest energy. In this phase the   particles become strongly correlated with a distributions both  in coordinate and momentum variables. The correlation between particles happend when  their thermal wave-lengths $\lambda$ become larger  than $\lambda_c$ which  is a function of the mass of  an individual condensate particle , the number density $n=\frac{1}{v}=\frac{N}{V}$($N$ is the number and $V$ is volume) and $k_B$ as Boltzmann’s constant \cite{Dalfovo}-\cite{Griffin}. Later it was discovered that BEC can emerge in molecular form in a Fermi gas \cite{Greiner}. From the mathematical point of view , the BEC waves are solitons  waves with shape invariant forms \cite{Solitons}. The applications of BEC is not limited to statistical  physics; in relativistic cosmology one can build a dark matter model using BEC \cite{Bettoni:2013zma} as well as dark energy \cite{Das:2014agf}. Furthermore one can solve the Einstein field equations in general relativity with cylindrical sysmmetry with BEC matter as energy-momentum source and find cosmic string solutions \cite{Harko:2014pba}. In relativistic astrophysics, the general relativistic stars can be formed with BEC matter sources \cite{Chavanis:2011cz}. A remarkable discovery was by  Horwitz et al., in ref. \cite{Burakovsky:1996dv}, where they discovered a new relativistic high temperature BEC regime. According to Horwitz et al.,the mass spectrum of such a system was  bounded.\par
 In this letter I investigate a generalized energy density function for a dilute non interacting bosonic gas in the grand canonical ensemble theory of standard statistical mechanics. I show that density and pressure can be calculated in closed form in terms of a generalized function , named Lerch function (sometimes called as Hurwitz-Lerch-Phi function). This class of generalized functions considered as a well posed version of Riemann's zeta function  and has recently been reintroduced to the community in a series of works by Lagaria \&  Li.  in refs.\cite{Lagaria1}-\cite{Lagaria5}. As far as I know \cite{utube}, this rich family of special functions has not been  used in  physical problems especially when we see they are potentially related to the  Riemann hypothesis and can be used to find a solution to this old problem which has been listed as a "millennium-problem"\cite{millennium-problems}. As an interesting mathematical physics problem in this letter I explore the role of this Lerch functions in a realistic statistical problem relating to BEC. 

\section{BEC for an Exponential Density of States Function}
The grand partition function for the bosonic gas system given by
\begin{eqnarray}
&&q=\ln D=\frac{PV}{k_BT}=-\sum_{\epsilon}\ln\big(1-ze^{-\beta\epsilon}
\big)\label{q}
\end{eqnarray}
here  $z=e^{\beta\mu}$ is the fugacity of the gas and is related to the chemical potential $\mu$ ,$\beta=\frac{1}{k_BT}$ and $ze^{-\beta\epsilon}< 1,\ \ \forall\epsilon $.  Furthermore, for enough large volume $V$, the spectrum of the single-potential state is almost continuous on  $\mathcal{R}$, it is  that the single-potential state  of a boson in a cube is proportioal to $E_n \propto n^2 V^{-2/3}$; here $n$ denotes a collective set of quantum numbers $(n_x,n_y,n_z)$ and $V$ denotes the volume. We observe that:
\begin{eqnarray}
&&\lim _{V\to\infty}\Delta E_n=\lim _{V\to\infty}(E_{n+1}-E_n)\\&&\nonumber \propto\lim _{V\to\infty} \frac{2n+1}{V^{2/3}}=0
\end{eqnarray}
As a result, we can use the integral instead of summation ,
\begin{eqnarray}
&&\sum_{\epsilon}\Rightarrow \int d\epsilon 
\end{eqnarray}
Here we consider a bosonic system in which the density of state function $a(\epsilon)$ in the vicinity of a given energy $\epsilon$ given by
 \begin{eqnarray}
 &&a(\epsilon)d\epsilon=\frac{2\pi V}{h^3}(2m)^{3/2}\epsilon^{1/2}e^{-c\epsilon}d\epsilon\label{a}.
 \end{eqnarray}
Note that by substituing this  density function  (\ref{a}) into (\ref{q}) we  eventually give a weight to the ground state $\epsilon=0$  incorrectly, because  in the quantum mechanical approach , each non- degenerate single particle state in the system has a unity weight. We calculate the $\epsilon=0$ term to obtain:
\begin{widetext}
 \begin{eqnarray}
&&\frac{P}{kT}=-\frac{\ln(1-z)}{V}
-\frac{2\pi  (2m)^{3/2}}{h^3}\int_{0+}^{\infty}\epsilon^{1/2}e^{-c\epsilon}\big(1-ze^{-\beta\epsilon}
\big)
 d\epsilon\label{p}
 \end{eqnarray}
\end{widetext}
\begin{widetext}
	\begin{eqnarray}
	&&
\label{n11}
 \frac{N}{V}=\frac{1}{V}\frac{z}{1-z}+\frac{2\pi  (2m)^{3/2}}{h^3}\int_{0+}^{\infty}\frac{\epsilon^{1/2}e^{-c\epsilon}}{z^{-1}e^{\beta\epsilon}-1
 }
 d\epsilon.
\end{eqnarray}
\end{widetext}
for $z\ll 1$ , each of the last terms are of order $\mathcal{O}(\frac{1}{N})$. However if $z$ increases , at $z\to 1$ , $\frac{1}{V}\frac{z}{1-z}=\frac{N_0}{V}$ and $N_0\gg N$  is called Bose-Einstein condensation. To evalute Eqs. (\ref{p}.\ref{n11}) . We define $\tilde{c}=\frac{c}{\beta}$ and the thermal wavelength $\lambda=\frac{h}{\sqrt{2\pi m k T}}$. Furthmore our aim will be to rewrite all expressions in terms of the following generalized Bose-Einstein functions , which will be related to Lerch functions:
	\begin{eqnarray}
	&&\Phi(\nu,z,\tilde{c})=L_{\nu}(z)=\frac{1}{\Gamma(\nu)}\int_{0^{+}}^{\infty}\frac{x^{\nu-1}e^{-\tilde{c}x}}{z^{-1}e^x-1}dx\label{lerch}
	\end{eqnarray}
If we calculate the integral, we obtain:
	\begin{eqnarray}
&&\Phi(\nu,z,\tilde{c})=\sum_{n=1}^{\infty}\frac{z^n}{(\tilde{c}+n)^{\nu}}
\label{phi}
\end{eqnarray}
here $\Phi(\nu,z,\tilde{c})=z\Phi(\nu,z,\tilde{c}+1)$ is the Lerch transcendent  function or Lerch zeta function  \cite{Lagaria1}-\cite{Lagaria5}. Note that if $\tilde{c}=z=1$,the Lerch function (\ref{phi})
reduces to the Riemann zeta function :
	\begin{eqnarray}
&&\Phi(s,1,1)=\zeta(s)\equiv \sum_{n=0}^{\infty}\frac{1}{(n+1)^s}.\label{series}
\end{eqnarray}
Now we can calculate (\ref{p},\ref{n}) using (\ref{phi}) we obtain:
	\begin{eqnarray}
&&\frac{N-N_0}{V}=\frac{1}{\lambda^3}\Phi(\frac{3}{2},z,\tilde{c})
\label{n1}\\&&\frac{P}{V}=\frac{1}{\lambda^3}\int dz\Phi(\frac{3}{2},z,\tilde{c}+1)\label{p0}
.
\end{eqnarray}
We can find an equsation of state using (\ref{n1},\ref{p0}) , for $z\ll 1$ we can find an explicit equsation of state  as follows:

\begin{eqnarray}
&&P\approx\kappa_c \left(N-N_0\right)-\eta_c VT^{3/2}.
\end{eqnarray}
where $\kappa_c=\sqrt{2} \frac{(\tilde{c}+1)^{3/2} }{\sqrt{2} \tilde{c}^{3/2} },\eta_c=4 \left(m  \pi 
k_B\right){}^{3/2}\frac{(\tilde{c}+1)^{3/2} }{\sqrt{2} \tilde{c}^3 h^3}$.

Furthermore we obtain:
	\begin{eqnarray}
&&z\frac{\partial}{\partial z}\Big[\frac{PV}{kT}\lambda^3
\Big]=\lambda^3\Big(\frac{N-N_0}{V}
\Big)
\end{eqnarray}
We can now find the $T_c$ for our system. Let us rewrite equation
(\ref{n1}) in the following equivalent form:

	\begin{eqnarray}
&&\lambda^3\frac{N_0}{V}=\frac{\lambda^3}{v}-\Phi(\frac{3}{2},z,\tilde{c})
\label{n2}
.
\end{eqnarray}
where $v=\frac{V}{N}$ the specific volume . 
The  figure (\ref{f0}) below shows that for $z\in[0,1],\tilde{c}\in\mathcal{R}^{+}$ , $\Phi(\frac{3}{2},z,\tilde{c})$ is a bounded, positive, monotonically increasing single valued function of $z$.

\begin{figure}
	\includegraphics[width=\linewidth]{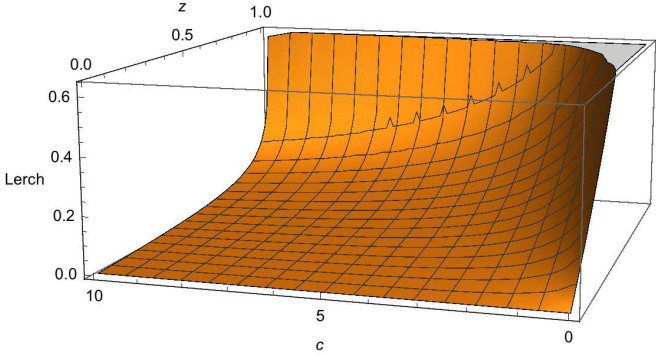}
	\caption{A plot of the $\Phi(\frac{3}{2},z,\tilde{c})$ of the Lerch transcendent  function given in equation (\ref{phi})  for $z\in[0,1],\tilde{c}\in\mathcal{R}^{+}$.
	It shows that $\Phi(\frac{3}{2},z,\tilde{c})$ is a bound, positive, monotonically increasing single valued function of $z$. }
	\label{f0}
\end{figure}

For small $z$, we have the power series (\ref{series}). At $z=1$ and for $\tilde{c}\in\mathcal{R}^{+}$, its value is finite and is called  Hurwitz-Zeta function :
	\begin{eqnarray}
&&\Phi(\frac{3}{2},1,\tilde{c})= \sum_{n=0}^{\infty}\frac{1}{(n+\tilde{c})^{\frac{3}{2}}}\leq\zeta(\frac{3}{2})\approx 2.612\label{series-z1}
\end{eqnarray}
Thus for all $z,\tilde{c}\in\mathcal{R}$, we conclude that $\frac{N_0}{V}>0$ when the temperature and the specific volume satify the following inequality:
\begin{eqnarray}
&&\frac{\lambda^3}{v}>\Phi(\frac{3}{2},1,\tilde{c})
\end{eqnarray}
This implies BEC phenomena, the occupiation of ground state with many particle. For a specific volume $v$, we define "implicitly" a critical temperature $T_c$ :

\begin{eqnarray}
&&kT_c=\frac{2\pi\hbar^2}{m\Phi(\frac{3}{2},1,kT_c c)^{2/3}}\label{tc}
\end{eqnarray}

For $ckT_c\ll1$ we can solve (\ref{tc}) and find
\begin{eqnarray}
&&ckT_c\approx 0.65 + 0.25m^{-1/2} \sqrt{6.8 m - 50.6 c \hbar^2}
\label{tc0}
\end{eqnarray}
for a lower bound on the Boson mass as $m\geq 7.4  c \hbar^2 $. Such lower bound on the mass of the  Boson (Higgs) proposed in \cite{Willey:1995zm}.

Simple algebraic manipulations can be done to show that  the quantity $\frac{N_0}{V}$  is a simple two-part function of $\frac{T}{T_c}$ :
\begin{eqnarray}
\frac{N_0}{V} =
\left\{
\begin{array}{ll}
0  & \mbox{if } \frac{\lambda^3}{v}<\Phi(\frac{3}{2},1,\tilde{c}) \\
1-(\frac{T}{T_c})^{\Delta} & \mbox{if } \frac{\lambda^3}{v}>\Phi(\frac{3}{2},1,\tilde{c})
\end{array}
\right.
\end{eqnarray}
where, to find the critical exponent $\Delta$ , we used numerical estimation in equation (\ref{tc}).

\section{More on density of states function}
There is a nice relationship between the exponential density of state function given in (\ref{a}) and momentum  of the particles $\vec{p}$ , is given as the relation between volumes in configuration space $V_{\vec{q}}$ and phase space $V_{\vec{q},\vec{p}}$ as follows,
	\begin{eqnarray}
&&a(\epsilon)=g\frac{dV_{\vec{p}}}{d\epsilon}\frac{V_{\vec{q}}}{V_{\vec{q},\vec{p}}}
\end{eqnarray}
where $g$ denotes  the number of  degrees of freedom for the bosonic  system  $g=3$. By
integrating this expression we find an \emph{atypical} dispersion relation as following:
	\begin{eqnarray}
&&\epsilon(p)=-c^{-1}\Big(1+PLog\Big[\alpha p^3+\gamma
\Big]
\Big)\label{e}
\end{eqnarray}
where $p\equiv |\vec{p}|>0$. By $PLog\Big[z\Big]$ we mean  the principal solution for $w\in\mathcal{C}$ in equation $z=we^w$ and $\alpha=\frac{\sqrt{2} c^2 }{3 e
	m^{3/2}},\gamma=\frac{\alpha  c^2 \hbar
	^3}{2 \sqrt{2} e m^{3/2} V \pi }$. At low momenta $p\ll \tilde{p}(=\alpha^{-1/3})\ll 1$ we have
	\begin{eqnarray}
&&\epsilon(\vec{p})\approx \epsilon_0+\gamma p^3+\mathcal{O}(p^4)
\end{eqnarray}
For ultra-energetic particles $p\to \infty$ we have:
	\begin{eqnarray}
&&\epsilon(p)\approx -c^{-1}\Big[(\ln p)^3+...
\Big]
\end{eqnarray}
A remarkable observation is that for some values of $\alpha,\gamma$  we can obtain $\epsilon(p^{*})=0,\ \  p^{*}\neq 0$. If we solve equation (\ref{e}) we find:
	\begin{eqnarray}
&&p^{*}=\sqrt[3]{\frac{e \gamma
		+1}{-\alpha}}.\ \ \alpha<0,\gamma>0.
\end{eqnarray}

This defines a minimum momentum (effective mass) similar to that given in \cite{Burakovsky:1996dv}, where the authors discovered  
a new relativistic high temperature BEC sytem with  an additional mass portential for the ensemble. In our system
 we can demonstrate  that there is no physical state with $p<p^{*}$ and the system has non zero energy even at zero momentum when $p=0$ the energy of the particle  will be a monotonically increasing function of $p$ for all $p>p^{*}$. 
A Theorem can be stated as follows:

\begin{theorem}
	Let $\epsilon(p)$ be energy spectrum in our atypical system , then $\epsilon(p)$ 
	is a complex valued function for $p<p^{*}$.\label{pythagorean}
\end{theorem}

And a consequence of theorem \ref{pythagorean} is the  following corollary:

\begin{corollary}
	 There is no  physical state with $p<p^{*}$.
\end{corollary}

To prove this theorem \ref{pythagorean} let us prove the following lemma:
\begin{lemma}
	Given a physical satate with $p=p^{*}-\delta,\ \ 0<\delta\ll p^{*} $,  there is a 
complex  number $W\in\mathcal{C}$ such that $\epsilon(p)=W$.
\end{lemma}
\begin{proof}
		to prove, we need to rewrite the energy spectrum given in equation (\ref{e}) for $p=p^{*}(1-\eta),\ \ |\eta|\ll 1 $ becomes a complex function. Let us
	first rewrite energy spectrum near $p^*$ in the following form:
	
		\begin{eqnarray}
	&&\epsilon(p)=-c^{-1}\Big(1+PLog\Big[-e^{-1}+\alpha(p^3-(p^{*})^3)
	\Big]
	\Big)\label{e1}
	\end{eqnarray}
We	expand the expression (\ref{e1}) as power series in $\eta$,
	 \begin{widetext}
		\begin{eqnarray}
	&&\epsilon(p)=(-1)^{\left\lfloor B\right\rfloor }
	\left(-\frac{
		\sqrt{-6eq^3 \alpha }}{c}\eta^{1/2}+O\left(\eta
	^{3/2}\right)\right)\label{e1}
	\end{eqnarray} 
	 \end{widetext}
	 where $\lfloor x\rfloor$ is the greatest integer that is less than or equal to $x$ and $B=\frac{1}{2}-\frac{\arg (\eta
	 	)+\arg \left(-3\alpha(p^{*})^3 \right) }{2 \pi
	 }$. Note that $\alpha<0$ and $|\eta|\ll 1$, consequently $\lfloor \eta\rfloor=0$. Thus we can rewrite the above series as follows:
	 	 \begin{widetext}
	 	\begin{eqnarray}
	 	&&\epsilon(p)=(-1)^{\left\lfloor \tilde{B}\right\rfloor }
	 	\left(-\frac{
	 		\sqrt{-6e(p^{*})^3 \alpha }}{c}\eta^{1/2}+O\left(\eta
	 	^{3/2}\right)\right)\label{e2}
	 	\end{eqnarray} 
	 \end{widetext}
	where $\tilde{B}=\frac{1}{2}-\frac{\arg \left(3(1+\gamma) \right) }{2 \pi
	}$. Obviously $(-1)^{\left\lfloor \frac{1}{2}-\frac{\arg \left(3(1+\gamma) \right) }{2 \pi
	 	}\right\rfloor }=1, \forall \gamma\in\mathcal{R}$ as a result of $\eta\to 0^{-}$, $\lim\epsilon(p)_{\eta\to 0^{-}}=ib+O\left(\eta
 	^{3/2}\right)\in\mathcal{C}$. This proves our lemma.
\end{proof}

\section{On "Riemann hypothesis" problem and classical BEC}
In this section I will show that how one can find a Hurwitz-Zeta function with complex argument  $\Phi(a+ib,1,\tilde{c})$ using the following density of  states  function: 
\[
a(\epsilon)=A\epsilon^{a-1}e^{-c\epsilon}\left\{
\begin{array}{ll}
\cos(b\ln \epsilon)\\
\sin(b\ln \epsilon)
\end{array}
\right.
\]
\label{a-Riemann}
for $a,b,c\in\mathcal{R}$.
Then by taking the limit  $c\to0$, we are led to Zeta function $\zeta(a+ib)$. The Riemann hypothesis  (stated below) will translate to a problem of finding specific values for $T_c\gg 1$. Let us use  density (\ref{a-Riemann}) to evaluate the grand partition function (\ref{q}). It is adequate to rewrite (\ref{a-Riemann})  as follows:
\begin{eqnarray}
&&a(\epsilon)=A\Re\{\epsilon^{a-1+ib}e^{-c\epsilon}\},\  \ A,a,b\in\mathcal{R}\label{a-Riemann2}
\end{eqnarray}
We can find an $A$ such that 
 $\lim_{c\to 0,a\to\frac{3}{2},b\to 0} \Big[a(\epsilon)\Big]=\frac{2\pi V}{h^3}(2m)^{3/2}$,
by plugging it into  the integral (\ref{lerch}) we obtain:
\begin{eqnarray}
&&\frac{N-N_0}{V}=\frac{A_0}{\lambda^3}\Re\{\Phi(a+ib,z,\tilde{c})\}
\label{p1}\\&&\frac{P}{V}=\frac{A_0}{\lambda^3}\int \frac{dz}{z}\Re\{\Phi(a+ib,z,\tilde{c})\}\label{nr}
.
\end{eqnarray}
where $A=A_0\Big[\frac{2\pi V}{h^3}(2m)^{3/2}\Big]$  and $A_0$ is a constant. The BEC scenario can be modified to finding $T_c\to 0$ using the following modified equation as a limiting case for $c\to0$:

\begin{eqnarray}
&&kT_c=\frac{2\pi\hbar^2}{mA_0(\Re\{\Phi(a+ib,1)\})^{2/3}}\label{tc2}
\end{eqnarray}
We can rewrite the above equation in the following equivalent form:
\begin{eqnarray}
&&\Re\{\zeta(a+ib)\}=
\Big(\frac{2\pi\hbar^2}{mA_0kT_c}
\Big)^{3/2}\label{zeta1}
\end{eqnarray}
In Fig.(\ref{zeta}) we plot $\Re\{\zeta(a+ib)\}$ in the half-plane $a\geq1$. The region where $\Re\{\zeta(a+ib)\}<0$ investigated in ref. \cite{Juan}, where the authors proposed algorithms  to compute accurately the  supremum of $a\approx 1.19$ such that $\Re\{\zeta(a+ib)\}\leq 0$ for $b\in\mathcal{R}$. This supremum value of $a$ leads to the density function
\begin{eqnarray}
&&\sup_{a\in 1.19} \inf_{b\in \mathcal{R}}a(\epsilon)=A\Re\{\epsilon^{0.19+ib}e^{-c\epsilon}\},\  \ A,a,b\in\mathcal{R}\label{a-Riemann2}
\end{eqnarray}

\begin{figure}
	\includegraphics[width=\linewidth]{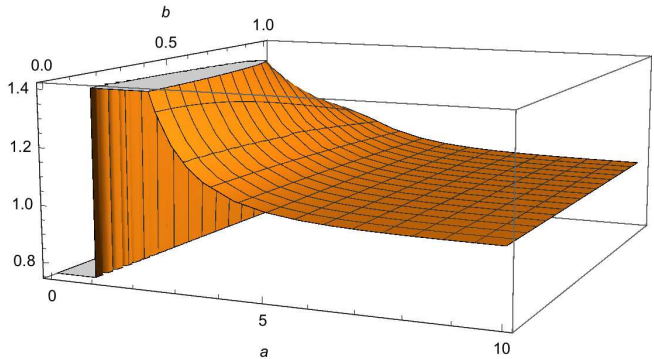}
	\caption{A plot of the real part $\Re\{\zeta(\nu)\}$  of the Riemann zeta-function given in equation (\ref{zeta1})  in the half-plane $\Re{\nu}\geq1$.   
	It shows that the zeros of $\Re\{\zeta(\nu)\}$ are located  only at  $(\Re{\nu}=0,0<\Im{\nu}<1)$ or $(\Im{\nu}=0,0<\Re{\nu}<1)$. }
	\label{zeta}
\end{figure}

We recall the Riemann  hypothesis  problem (stated below):
\begin{theorem}
Riemann hypothesis:
	The zeros of $\zeta(a+ib)=0$ exist  only at $(\nu=-n,n\in\mathcal{R})\lor (\Re{\nu}=\frac{1}{2},\nu\in\mathcal{C})$.
\label{Riemann hypothesis}
\end{theorem}

\begin{lemma}
In BEC scenario with the critical temperature $T_c$ given in the equation 
(\ref{zeta1}),we have to check whether there exists a phase  transition in the classical  limit $\hbar\to 0$(or high critical temperatures $T_c$) of bosnonic matter with density of states function (\ref{a-Riemann2})  with  $a=\frac{1}{2} $ or not.
\end{lemma}
The above lemma implicitly encourages us to find a classical limit of BEC.  Such a classical limit was well investigated before in the literature \cite{Staliunas}-\cite{Davidson:2014hfa}. To find the classical analog of BEC, we need to have a relativistic classical Gibbs ensemble theory . The  Gibbs ensembles in relativistic classical and quantum mechanics was investigated in \cite{Horwitz:1980np}. The existence of a classical  BEC state depends on how to realize classical condensate state as a  macroscopic matter wave. The fingerprints of quantum matter waves was observed in \cite{Ketterle} where the authors observed the  transition temperature $T_c$ which is lower than  the thermodynamic limit. In ref. \cite{Staliunas}, the author suggested it is possible to have  BEC as a coherent dynamics of the system.
There is a chance of detecting axions as classical fields for BEC at the large scales proposed in \cite{Davidson:2014hfa} where the author suggested that the axions as a motivated dark matter candidate can be consider the  BEC. The classical limit of BEC can be underestood when we consider condensate field as a classical field , in analogy  to the classical limit of quantum optics. If we have many quanta in each mode , then the fields are treated as the classical fields. As next step we can use the classical fields approximation. We will  need  two point correlation function . These correlation functions  can depend on the resolution of detectors in an appropriate experimental set up (coarse graining) and we can use the Onsager-Penrose definition of the condensate\cite{Onsager-Penrose}. I conclude that further studies about classical limits for BEC can be useful to find a general formula for zeros of zeta function on complex plane. My idea here should be considered
as a physical proposal to mkae a connection between a pure mathrematical hypothesis and the realistic case of physical situation in nature. Before doing this calcuolaion, we were not aware of the connection between zeta function with complex argument and the statistical behgavior of a gaseous system. The difficulty was to set the complex argument in the zeta function according to the parameters of statistical system. I show how  the real part of the zeta function can be related to the critical temperature  and it establishes a nice connection between pure mathematics and applied physics.
\section{Creation-annihilation operators  for BEC}

In this section, I study  further field theoretic aspects of such an atypical \textbf{such atypical} condensation phenomenon in terms of the creation(raising) and annihilation(lowering) operators proposed for Hurwitz-Lerch-Phi function by Lagaria and  Li.  in refs.\cite{Lagaria1}-\cite{Lagaria5}. My focus is to explore the hidden physics of these operators in the mechanism of BEC in our system. In my study I found that the Lerch-zeta function $\Phi(a+ib,z,\tilde{c})$ plays an essential role. The pair of differential-difference operators $\frac{\partial}{\partial \tilde{c}}, z \frac{\partial}{\partial z}$ can be used to define a pair of  creation(raising)$D_{L}^{\pm}$ and annihilation(lowering) operators as follows:
\begin{eqnarray}
&&D_L^{+}=\frac{\partial}{\partial \tilde{c}}\label{d1}\\&&
	D_L^{-}=\frac{\partial}{\partial \ln z}+\tilde{c}\label{d2}.
\end{eqnarray}
where $[D_L^{+},D_L^{-}]_{-}=1$.
Note that these operators are non-self-adjoint linear operators acting in a Hilbert spaces. The reason is that the following integral identity no longer is valid:
\begin{widetext}
\begin{eqnarray}&&
\int_{\partial\Box}{\psi}^{\dagger}(z,\tilde{c})D_L^{\pm}\phi(z,\tilde{c})dzd\tilde{c}=\int_{\partial\Box}(D_L^{\pm}\psi(z,\tilde{c}))^{\dagger}\phi(z,\tilde{c})dzd\tilde{c}.
\end{eqnarray}
\end{widetext}
where $\partial\Box=[0,1]\times[0,1]$.
The corresponding eigenvalues of these operators are not real numbers and their corresponding wave functions can't built an orthonormal basis for the Hilbert space 
which should  be build for any self sdjoint operator.
The Lerch differential number operator is defined as follows:
\begin{eqnarray}
&&D_L\equiv D_L^{-} D_L^{+}=\frac{1}{\tilde{c}}
\frac{\partial^2}{\partial \ln z\partial \ln\tilde{c}}+\frac{\partial}{\partial \ln\tilde{c}}.
\end{eqnarray}
here $[D_L,D_L^{+}]_{-}=-D_L^{+},[D_L,D_L^{-}]_{-}=D_L^{-}$. The Hamiltonian like operator can be written as $H=1+D_L$.
It is illustrative to show that 
\begin{eqnarray}
&&D_L^{+}\Phi(\nu,z,\tilde{c})=-\nu\Phi(\nu+1,z,\tilde{c})
\label{d11}\\&&
D_L^{-}\Phi(\nu,z,\tilde{c})=\Phi(\nu-1,z,\tilde{c})\label{d22}.
\end{eqnarray}
combining equations (\ref{d11},\ref{d22}) we conclude that 
\begin{eqnarray}
&&D_L\Phi(\nu,z,\tilde{c})=-\nu\Phi(\nu,z,\tilde{c})
\label{d111}.
\end{eqnarray}
We  interpret the
physical operator equations as creation and annihilation equations to make create and annihilate a hypothetical virtual Boson which we call "Lerchton". This pair production phenomena happens in a phase space spanned by two coordinates $(\ln z,\tilde{c})$. The creation operator $D_L^{+}$ creates a Lerchton at fixed chemical potential and $D_L^{-}$ annihilates it but after transforming it to the new location, i.e. $\tilde{c}+1$ in the phase space. We mention here  that $\ln z\sim \mu $
defines an energy scale and $\tilde{c}\sim\beta\sim E^{-1}$ corresponds to a length scale. The first variable needs an ultraviolet cutoff and the length scale required an infrared cutoff. Using equation (\ref{d111}) we are able to interpret $-\nu$ as the expectation value of the number operator $D_L$(similar to the number operator  in a simple harmonic oscillator in ordinary qunatum mechanics) at a given bipartite state  defined by  $|\nu,\tilde{c}>$. The Lerch function  $\Phi(\nu,z,\tilde{c})\equiv <z|\nu,\tilde{c}>$, is the configurational projection of the original ket vector. An alternative form to the equation (\ref{d111}) is:
\begin{eqnarray}
&&D_L|\nu,\tilde{c}>=-\nu|\nu,\tilde{c}>
\label{d1110}.
\end{eqnarray}
Define $|\nu,\tilde{c}>$   as the normalized eigenstates of $D_L$ , and let it be understood that the physical states are labeled by the eigenvalue, i.e. $-\nu\in\mathcal{C}$ is complex in general. Note that at negative integer roots of zeta function $\zeta(\nu=-n)=0$, the eigenvalue of the operator $D_L$ is positive and it correponds to a real number of particles. 
If we consider the quantity \begin{eqnarray}
&&n\equiv <\nu,\tilde{c}|D_L|\nu,\tilde{c}>\end{eqnarray} 

It follows that 
\begin{eqnarray}
&&n= \frac{(\zeta(\nu)-1)}{1-\nu^{-1}} \label{n}
\end{eqnarray}
 is real and positive -definite for $\nu\in\mathcal{R}^{+}$. We plot the number of particles (\ref{n}) in figure (\ref{fign}) . 

\begin{figure}
	\includegraphics[width=\linewidth]{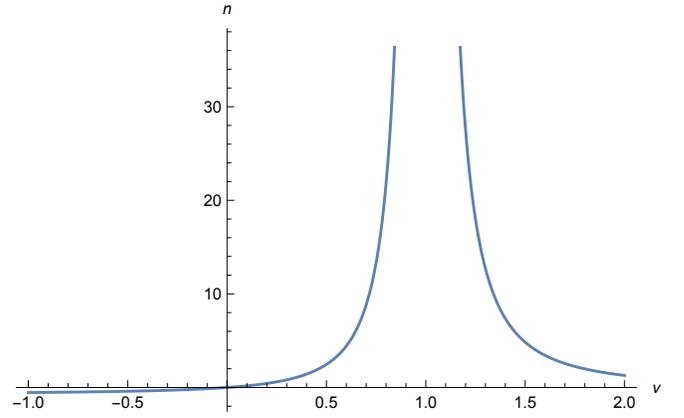}
	\caption{A plot of the number of perticles given in equation (\ref{n})  in the plane $-1<\Re{\nu}\leq2$   
		. It shows that the number of particles is real and positive -definite for $\nu\in\mathcal{R}^{+}$. It   infinite when $\nu\to1$. When $\nu\to 0$ the number of particles vanishes and it implies BEC in the state there isn't any particle at this level.    }
	\label{fign}
\end{figure}
From figure (\ref{fign}) we learn that the number of the particles become infinity when $\nu\to 1$. It implies an existed BEC condensation at state $|1,\tilde{c}>$ as the ground state of the system. Fuurther works can be done to build the Fock space corresponds to the multiparticle states. \par
For its intrinsic interest we now study coherent states correspondingh to the annihilation operator $D_L^{-}$, as the solution to the following differential-difference equation:
\begin{eqnarray}
&&
D_L^{-}\sum_{\nu}C_{\nu} \Phi(\nu,z,\tilde{c})=\alpha \sum_{\nu}C_{\nu} \Phi(\nu,z,\tilde{c})\label{coh}.
\end{eqnarray}
Using equation (\ref{d22}) we can tranform it to the folllowing equation:
 \begin{eqnarray}
 &&
 \sum_{\nu}(C_{\nu+1}-\alpha C_{\nu}) \Phi(\nu,z,\tilde{c})+\frac{C_0z}{1-z}=0\label{coh1}.
 \end{eqnarray}
Note that here $\sum_{\nu}$ is a symbol for a more general measure $\int d\nu$ over the full complex plane. The equation (\ref{coh1})is a differential-difference equation should be solved for $C_{\nu}$. Note that $ \Phi(\nu,z,\tilde{c})$ are considered as a set of linearly independent functions because they are eigenfunctions for the general operator $D_L$. A possible ad-hoc solution to equation (\ref{coh1}) is  given for $C_0=0$ and the coherent state then corresponds to it is expressed   as eq. (\ref{coh12}):
 \begin{eqnarray}
&&|\alpha>=A_0\sum_{\nu}\alpha^{\nu-1} \Phi(\nu,z,\tilde{c})
\label{coh12}.
\end{eqnarray}
We stress here that the above coherent state isn't normalized, but we can make it normalized if we opt
 \begin{eqnarray}
&&A_0=\Big(\sum _{\nu =-i\infty }^{i\infty } \frac{  \zeta (\nu )-1}{1-\nu^{-1} }\alpha ^{2 (\nu -1)}
\Big)^{-1/2}
\end{eqnarray}
where we have taken into account all values of $\nu\in\mathcal{C}$.

\subsection{Finding minimum of $\nu$ via variational method} 
The eigenvalue-eigenfunction operator equation given in (\ref{d111}) provides an easy way to find minimum of $\nu$ through a useful variational method. It is necessary to introduce a suitable functional on the phase space $(\ln z,\tilde{c})$ in the form that the resulting Euler-Lagrange equation gives us the linear second order partial differential equation proposed in 
The eigenvalue-eigenfunction operator introduced  given in (\ref{d111}). A possible functional to minimize for a trial function $\psi$ is given as follows
\begin{widetext}
	\begin{eqnarray}
	&&\Re \nu_{0}=-\Re Minimize\Big[\frac{\int_{\partial\Box}\bar{\psi}(z,\tilde{c}) D_L^{-} D_L^{+}\psi(z,\tilde{c}) dzd\tilde{c}}{\int_{\partial\Box}|\psi(z,\tilde{c})|^2dzd\tilde{c}}\Big]
\label{nu0}
	\end{eqnarray}
\end{widetext}
It defines a single-valued analytic functional to the region $\partial\Box$ and 
 the trial function $\psi$ should satisfy the essential boundary condition: 
\begin{widetext}
	\begin{eqnarray}
	&&\psi(z,0)=
	\left\{
	\begin{array}{ll}
	0  & \mbox{if } z\to 0 \\
1
	&  \mbox{if } z\to 1
	\end{array}
	\right.
	\end{eqnarray}
\end{widetext}
A possible trial function could be $\psi(z,\tilde{c})=(1+\tilde{c})z$. We evaluate (\ref{nu0}) we obtain $\Re \nu_{0}=-1$.

\section{summary}
The special functions of mathematical physics play essential roles in building new physics and developing new  mathematical tools in modern physics. The Hurwitz-Lerch-Phi function proposed as an attempt to generalize the Riemann zeta function to the complex plane as well as break its non simplicity to write it as a solution for a linear partial differential equation; the last was Riemann's attempt to zeta function. The Riemann zeta functions appeared as key functions in the study of the statistical mechanics of on a non interacting Boson dilute gas in the grand canonical ensemble. So far to my knowledge there isn't any other physical system with Hurwitz-Lerch-Phi functions as their key functions. I demonatrated here that  statistical mechanics for an expoential density of states functions leads to Hurwitz-Lerch-Phi function. if we set  the Lerch parameter to zero , it was shown that the statistical expressions for number density (specific volume) and pressure can be represented as the real part of the Hurwitz-Lerch-Phi function. A remarkable observation was that the bosonic system doesn't exist for momentum below a specific momentum $p^{*} $, where below this value the energy spectrum becomes complex and it will break the reality of our system. This statement is proved in detail by considering the left branch of the energy spectrum function. Furthermore I showed that there is a real valued density of states function with a supremum value over the whole energy range. I proposed a softer theorem in support of the idea that the roots of the real part of the zeta function in the half-plane of the complex argument can be obtained in general closed form by investigating classical limits of condensation phenomena. I mention here that such classical limits can be understood by two ways: one is by passing to $\hbar\to0$, and secondly by considering a very hot bath in the condensed phase. A  field theoretic point of view was introduced to interpret the Lerch functions as normalized stste for an irregular type of number operator in Fock space. I demonstrated that the number of particles in a given state can be a positive number on the half-complex plane of the zeta function parameter $\nu$. This provides a systematic basis to study multiparticle states for such BEC system from a field theory point of view. In a forthcoming paper I will study more field theoretical aspects of this system. The aim will be to prove the Riemann hypothesis in  in a weaker form that  proposed in this paper.

\section{Acknowledgments}
I thank Prof. Jeffrey. C. Lagarias (University of Michigan) for  very useful comments and discussions. 



\end{document}